                \documentclass{article}  % Comment this line out

\usepackage{amsmath,amsthm,amsfonts}
\usepackage{amssymb}
\usepackage{graphicx}
\usepackage{setspace}
\newtheorem{subthm}[subsubsection]{Theorem}
\newtheorem{prop}[subsection]{Proposition}

\newtheorem{subcoro}[subsubsection]{Corollary}
\newtheorem{subprop}[subsubsection]{Proposition}
\newcommand{\R}{{\mathbb{R}}}

\newtheorem{lem}[subsubsection]{Lemma}

\newtheorem{examples}[subsection]{Example}

\title{
A Constrained Evolutionary Gaussian Multiple Access Channel Game}

\author{Quanyan Zhu, Hamidou Tembine, Tamer Ba\c{s}ar}
\begin{document}

\maketitle
\thispagestyle{empty}
\pagestyle{empty}

%%%%%%%%%%%%%%%%%%%%%%%%%%%%%%%%%%%%%%%%%%%%%%%%%%%%%%%%%%%%%%%%%%%%%%%%%%%%%%%%
\begin{abstract}
In this paper, we formulate an evolutionary multiple access channel game with continuous-variable actions and coupled rate constraints. We characterize Nash equilibria of the game and show that  the pure Nash equilibria are Pareto optimal and  also resilient to deviations by coalitions of any size, i.e., they are strong equilibria. We use  the concepts of  price of anarchy and strong price of anarchy to  study the performance of the system. The paper also addresses how to select one specific equilibrium solution using the concepts of normalized equilibrium and evolutionary stable strategies. We examine the long-run behavior of these strategies under  several classes of evolutionary game dynamics such as Brown-von Neumann-Nash dynamics, and replicator dynamics.\footnote{Q. Zhu and T. Ba\c{s}ar are with the Department of Electrical and Computer Engineering and the Coordinated Science Laboratory,
University of Illinois at Urbana-Champaign. Postal Address: 1308 West Main, Urbana, IL, 61801, USA.
E-mail:\{zhu31,tbasar\}@decision.csl.uiuc.edu;
H. Tembine is with LIA/CERI, University of Avignon, France. E-mail: hamidou.tembine@univ-avignon.fr} \footnote{This work was done when the second coauthor was visiting University of Illinois at Urbana Champaign. This work was partially supported by an INRIA PhD intership grant.}
\end{abstract}
%\begin{IEEEkeywords}AWGN, Evolutionary Games.\end{IEEEkeywords}

\section{Introduction}

Recently, there has been much interest in understanding the behavior
of multiple access channels under constraints. Considerable amount of work has been carried out on the problem of how users can obtain an acceptable throughput by choosing  rates independently.
 Motivated by the interest in studying a large population of users playing the game over time, evolutionary game
theory was found to be an appropriate framework for communication networks. It has been applied to  problems such as power control in wireless networks  and mobile interference control \cite{networking}.
%In \cite{mass}, a medium access game with delay  of transmission has been studied  using bio-inspired game dynamics. The authors show that the medium access game has a unique evolutionary stable strategy and the  replicator dynamics converges to evolutionary stable strategy with small time delays.
 In \cite{pairwise}, an additive white Gaussian noise (AWGN) multiple access channel problem was modeled
 as a noncooperative game with pairwise interactions,
in which users were modeled as rational entities whose only
interest was to maximize their own communication rates. The authors obtained the  Nash equilibria of the two-user
game and introduced a two-player evolutionary game model with {\it pairwise interactions} based on replicator dynamics. However,  the case when interactions   are not pairwise arises frequently  in communication networks, such the  Code Division Multiple Access (CDMA) or the Orthogonal Frequency-Division Multiple Access (OFDMA) in Worldwide Interoperability for Microwave Access (WiMAX) environment \cite{networking}.

In this  work, we extend the study of \cite{pairwise} to wireless communication systems with an arbitrary number of users corresponding to each receiver. We formulate a static non-cooperative game  with $m$ users subject to rate capacity constraints  and  extend the constrained game to a dynamic evolutionary game with a large number of users whose strategies evolve  over time.
Different from evolutionary games with discrete and finite number of actions, our model is based on a class of  continuous games, known as  {\it continuous-trait games}.  Evolutionary games with continuum action spaces can be seen in a wide variety of applications
in evolutionary ecology, such as evolution of phenology, germination, nutrient
foraging in plants, and predator-prey foraging \cite{continuous,vincent05}.

\subsection{Contributions}
The main contributions of this work can be summarized as follows.
We show that the static continuous kernel rate allocation game with coupled rate constraints has a convex set of pure Nash equilibria, coinciding with the maximal face of the polyhedral capacity region.
All the pure equilibria are Pareto optimal and are also strong equilibria, resilient to simultaneous deviation by coalition of any size.
 We show that the pure Nash equilibria in the rate allocation problem are 100\% efficient in terms of   Price of Anarchy (PoA) and constrained Strong Price of Anarchy (CSPoA).
We study the stability of strong equilibria, normalized equilibria, and evolutionary stable strategies (ESS) using evolutionary game dynamics such as  Brown-von Neumann-Nash dynamics, generalized Smith dynamics, and replicator dynamics.

\subsection{Organization of the paper}
The rest of the paper is structured as follows. We present in the next section the evolutionary game model of rate allocation in additive white Gaussian multiple access wireless networks, and analyze its equilibria and Pareto optimality. In Section~\ref{secrobust}, we present strong equilibria and price of anarchy of the game. In Section \ref{secselection}, we discuss how to select one specific equilibrium such as normalized equilibrium and evolutionary stable strategies. Section~\ref{secdynamics} studies the stability of equilibria and evolution of strategies using game dynamics. Section~\ref{secconclud} concludes the paper.

\section{The Game Model } \label{secmodel}

We consider a communication system consisting of several receivers and several senders (See Figure \ref{figfuncttt3}). At each time, there are many local interactions (typically, at each receiver there is a local interaction) at the same time. Each local interaction will correspond to a non-cooperative  one-shot game with common constraints. The opponents do not necessarily stay the same from a given time slot to another time slot. Users revise their rates in view of their payoffs and the coupled constraints (for example by using an evolutionary process, a learning process or a trial-and-error updating process). The game evolves  in time.
Users are  interested in maximizing a fitness function based on
their own communication rates at each time, and  they are aware of the fact that
the other users have the same goal.
The coupled power and rate constraints are
also common knowledge. Users have to choose independently
their own coding rates at the beginning of the communication,
where the rates selected by a user may be either deterministic,
or chosen from some distribution. If the rate profile arrived at as a result of these independent decisions lies
in the capacity region, users will communicate at that operating
point. Otherwise, either the receiver is unable to decode any
signal and the observed rates are  zero, or only one of the
signals can be decoded. The latter case occurs when  all the other users are transmitting at
or below a safe rate. With these assumptions, we can define
a constrained non-cooperative game. The set of allowed strategies for user $j$ is the set
of all probability distributions over $[0,+\infty[,$ and the payoff is a function of
the  rates. In addition, the rational action (rates) sets are restricted to lie in the capacity regions (the payoff is zero if the constraint is violated).
In order to study the interactions between the selfish or partially cooperative users
and their stationary rates in the long run,
we propose to model the rate allocation in Gaussian multiple access channels as an
evolutionary game with a continuous action space and coupled constraints. The development
of evolutionary game theory is a major contribution of
biology to competitive decision making and the evolution of cooperation. The key concepts of
evolutionary game theory are
 (i) {\it Evolutionary Stable Strategies} \cite{smith},  which is a refinement of equilibria, and (ii) {\it Evolutionary Game Dynamics}  such as replicator dynamics \cite{taylor}, which describes
the evolution of strategies or frequencies of use of strategies in time, \cite{vincent05,hofbauer}.
\begin{figure}[htb]
\centering
\includegraphics[scale=0.45]{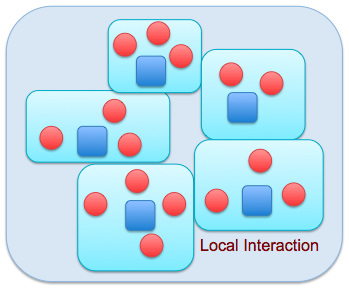}
\caption{A population: distributed receivers and senders, represented by blue rectangles and red circles respectively.} \label{figfuncttt3}
\hspace{0.2cm}
\end{figure}

The single population evolutionary rate allocation game is described as follows:
there is one population of senders (users) and several receivers. The number of senders is  large. At  each time, there are many one-shot games called {\it local interactions}.
Each sender of the
population chooses from the same set of strategies
${\mathcal{A}}$ which is a  non-empty, convex and compact subset of $\mathbb{R}.$ Without loss of generality, we can suppose that user $j$ chooses its rate in the interval $\mathcal{A}=[0,C_{\{j\}}]$, where  $C_{\{j\}}$ is the rate upper bound for user $j$ (to be made precise shortly), as outside of the capacity region the payoff (as to be defined later) will be zero.
Let $\Delta({\mathcal{A}}) $ be
the set of probability distributions over the  pure strategy set $\mathcal{A}.$ The set
$\Delta({\mathcal{A}})$ can be interpreted as the set of mixed strategies.
It is also interpreted as the set of distributions of strategies among
the population. Let $\lambda_t\in \Delta({\mathcal{A}}),$ and $E$ be a   $\lambda_t-$ measurable subset of $\mathbb{R}^m$; then  $\lambda_t(E)$ represents the fraction of users choosing
a strategy out of $ E$, at time $t.$ A distribution $\lambda_t\in \Delta({\mathcal{A}}) $ is sometimes called the ``state"
 of the population.   We denote by $\mathbb{B}(\mathcal{A})$ the Borel $\sigma-$algebra on ${\mathcal{A}}$ and by $d(\lambda,\lambda')$ the distance between two states measured with the respect to the weak topology.
 Each user's payoff depends on opponents' behavior through
the distribution of opponents' choices and of their strategies. The payoff of a user $j$ in a local interaction with $(m-1)$ other users is given as a function $u^j:\ \mathbb{R}^m\longrightarrow \mathbb{R}.$ The rate profile $\alpha\in\mathbb{R}^m$ must belong to a common capacity region $\mathcal{C}\subset\mathbb{R}^m $ defined by $2^m-1$ linear inequalities. The
expected payoff of a sender transmitting with the rate $a$ when the state of the
population is $\mu\in \Delta(\mathcal{A})$ is given by  $F(a,\mu).$  The expected payoff is $$F(\lambda,\mu):=\int_{\alpha\in \mathcal{C}}u(\alpha)\ \lambda(d\alpha^j) \prod_{i\neq j}\mu(d\alpha^i) .$$
The population state is subjected to the ``mixed extension" of capacity constraints $\mathcal{M}(\mathcal{C}).$
This will be discussed in Section \ref{secdynamics} and will be made more precise later.
%If $\lambda$ is the population profile, the set of rates in the capacity region are in $[0,C_{\Omega}-(m-1)\mathbb{E}(\lambda)]$

\subsection{Local Interactions}
%A local interaction consists a set of senders that transmit at a receiver.
A local interaction refers to the problem setting of one receiver and its uplink additive white Gaussian noise (AWGN) multiple access channel with several senders (say $m\geq 2$) with coupled constraints (or actions). The signal at the receiver is given by $ Y=\xi+\sum_{j=1}^{m}X_j$ where $X_j$ is a transmitted signal of user $j$ and $\xi$ is zero mean Gaussian noise with variance $\sigma_0^2.$ Each user has an individual power constraint $\mathbb{E}(X_j^2)\leq P.$  The optimal power allocation scheme
is to transmit at the maximum power available, i.e. $P$, for each user. Hence, we consider the case in which maximum power is attained. The decisions of the users  then  consist of choosing
their communication rates, and the receiver's  role is to
decode, if possible. The capacity region is a set of all vectors $\alpha\in \R^{m}_{+}$ such that users $j=1,2,\ldots,m$  can reliably communicate at rate $\alpha^j, ~j=1,\ldots,m.$
The capacity region $\mathcal{C}$ for this channel is the set
\begin{eqnarray} \label{tet}
\mathcal{C}=
\nonumber \left\{\alpha\in \R^{m}_{+} ~\bigg|~ \sum_{j\in J}\alpha^j\leq   \log\left(1+|J|\frac{P}{\sigma^2_0}\right), \right.\\
 \left. \forall\ \emptyset \subsetneqq J\subseteq \Omega\right\}
 \end{eqnarray}
 where $\Omega:=\{1,2,\ldots,m\}.$
 We refer the reader to \cite{ref2} for more details on the capacity region.
Notice that there is a tradeoff between high and low rates: if user $j$ wants to communicate at a higher rate, one of the other users $k$ may need to lower its rate, otherwise the capacity constraint is violated.

\begin{examples}{(Example of capacity region with three users)}
In this example, we illustrate the capacity region with three users. Let $\alpha^1,\alpha^2,\alpha^3$ be the rates of the users. Based on (\ref{tet}), we obtain

$$\left\{\begin{array}{l}\alpha^1\geq 0,\alpha^2\geq 0,\alpha^3\geq 0\\
\alpha^1\leq \log(1+\frac{P}{\sigma_0^2})\\
\alpha^2\leq \log(1+\frac{P}{\sigma_0^2})\\
\alpha^3\leq \log(1+\frac{P}{\sigma_0^2})\\
\alpha^1+\alpha^2\leq \log(1+2\frac{P}{\sigma_0^2})\\
\alpha^1+\alpha^3\leq \log(1+2\frac{P}{\sigma_0^2})\\
\alpha^2+\alpha^3\leq \log(1+2\frac{P}{\sigma_0^2})\\
\alpha^1+\alpha^2+\alpha^3\leq \log(1+3\frac{P}{\sigma_0^2})\\
\end{array}\right. \Longleftrightarrow M_3\gamma_3\leq \zeta_3,\ $$ where in the compact notation, $$ \gamma_3:=\left(\begin{array}{c}\alpha^1\\ \alpha^2\\ \alpha^3\end{array}\right)\in\mathbb{R}_+^3,\
\zeta_3:=\left(\begin{array}{c}C_{\{1\}}\\ C_{\{2\}}\\ C_{\{3\}}\\ C_{\{1,2\}}\\ C_{\{1,3\}}\\ C_{\{2,3\}}\\ C_{\{1,2,3\}}\end{array}\right) ,$$ $$M_3:=\left(\begin{array}{ccc}1 & 0 & 0 \\ 0 & 1 & 0\\ 0 & 0 & 1\\
 1 & 1 & 0\\ 1& 0 & 1\\ 0 & 1& 1\\ 1 & 1 & 1
 \end{array}\right)\in \mathbb{Z}^{7\times 3}.$$ Note that $M_3$ is a totally unimodular matrix.
By letting $P=25,\sigma_0^2=0.1,$ we show in Figure~\ref{figfunctt3} the capacity region with three  users.
\begin{figure}[htb]
\centering
\includegraphics[scale=0.4]{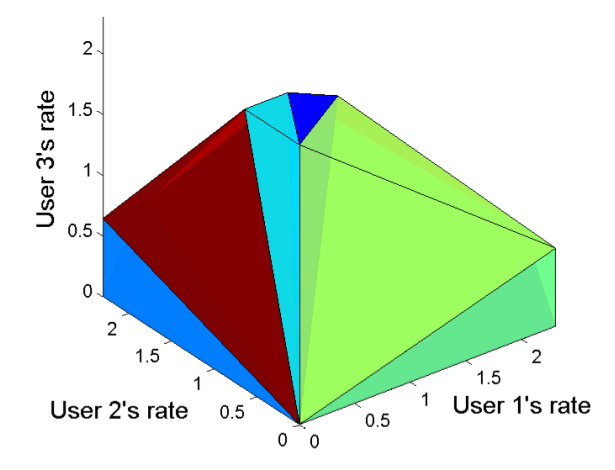}\\ %regiontt2
\caption{Capacity region with three users.}
  \label{figfunctt3}
\end{figure}
\end{examples}

We denote by $$r_{m}=\log\left(1+\frac{P}{\sigma_0^2+(m-1) P}\right)$$ the  rate  of a user when the signal of the $m-1$ other users is treated as  noise, and  $C_J=\log(1+|J|\frac{P}{\sigma_0^2}) $ its capacity. Note that $r_m=C_{\{m\}}-C_{\{m-1\}}.$
The set $\mathcal{C}$ is clearly a non-empty and bounded subset of $\R^{m}.$  $\mathcal{C}$ is closed and  is defined by $2^{m}-1$  convex inequalities. Thus, $\mathcal{C}$ is convex and compact.  From  the inequality
 $$\log\left(1+\sum_{j\in J}x_j\right)\leq \log\left(\prod_{j\in J}(1+x_j)\right)=\sum_{j\in J}\log(1+x_j),$$ for all $\forall x\in \R^{|J|}_{+},$ we obtain $C_{J}\leq \sum_{j\in J}C_{\{j\}}.$
%Let $J$ be a subset of $\Omega$ such that $|J|\geq 2$ and $i_0\in J.$ Taking the sum over $i$ of  the constraints $\sum_{j\in J\backslash \{i\}}\alpha^j\leq C_{J\backslash \{i\}},\forall i\in J,$  one has $ \sum_{j\in J}\alpha^j\leq \frac{|J|}{|J|-1}C_{J\backslash \{i_0\}}.$ We show that the $\frac{|J|}{|J|-1}C_{J\backslash \{i_0\}}$ is greater than $C_J.$ This condition is equivalent to \begin{eqnarray}\nonumber C_{J} & \geq  & |J|\left(C_J-C_{J\backslash \{i_0\}}\right)= \frac{|J|}{m}\log\left(1+\frac{P}{\sigma_0^2+(|J|-1)P}\right).\end{eqnarray}
%
%\begin{eqnarray} \nonumber
%C_J  \leq  \frac{|J|}{|J|-1}C_{J\backslash \{i_0\}} \leq   \frac{|J|(|J|-1)(|J|-2)\ldots 2}{(|J|-1)(|J|-2)\ldots \times 1}C_{\{1\}}  = \frac{|J|}{m}\log(1+\frac{P}{\sigma_0^2})
%\end{eqnarray}
%
%We conclude that the inequalities $\sum_{j\in J}\alpha^j\leq C_{J}$ are not redundant.

\subsection{Payoff}%0\leq\alpha^{j}\leq  r_{m}, \mbox{or}\
We define the payoff of user $j$ as $$ u^j(\alpha^{j},\alpha^{-j})=\left\{ \begin{array}{ll} g(\alpha^j) & \mbox{if}\   (\alpha^{j},\alpha^{-j})\in \mathcal{C}\\
0 & \mbox{otherwise}
\end{array}
\right.,$$ where $\alpha^j$ is the rate of the user $j$; the vector $\alpha^{-j}:=(\alpha^1,\ldots,\alpha^{j-1},\alpha^{j+1},\ldots,\alpha^m)$ is a profile of rates of the other users; the function $g:\ \R\rightarrow \R$ is a positive and strictly increasing  function. Given the strategy profile $\alpha^{-j}$ of the others players,  player $j$ has to maximize $u^j(\alpha^{j},\alpha^{-j})$ under its action constraints $$\mathcal{A}(\alpha^{-j}):=\{\alpha^{j}\in [0,C_{\{j\}}],\ (\alpha^{j},\alpha^{-j})\in \mathcal{C} \}.$$
Using the monotonicity of the function $g$ and the inequalities that define the capacity region, we obtain the following lemma.
\begin{lem} \label{lembr}
Let $\overline{BR}(\alpha^{-j})$ be the best reply to the strategy $\alpha^{-j}$ is  $$ \overline{BR}(\alpha^{-j})=\arg\max_{y\in \mathcal{A}(\alpha^{-j})}u^j(y,\alpha^{-j}).$$ $\overline{BR}$ is a non-empty single-valued correspondence (i.e., a  standard function) which is given by  $$\max\left(r_{m},\min_{J}\left\{ \ C_{J}-\sum_{k\in J\ \atop k\neq j}\alpha^{k},\ J\in \Gamma_j\right\}\right)\,$$ where $\Gamma_j :=\{J\in 2^{\Omega},\  J \ni j\}$.
\end{lem}
\begin{subprop}\label{ne}
The set of Nash equilibria is $$\{ (\alpha^j,\alpha^{-j})\ |\  \alpha^{j}\geq r_m,\sum_{j}\alpha^{j}=C_{\Omega}\}.$$ All these equilibria are optimal in the Pareto sense.\footnote{An allocation of payoffs is Pareto optimal or Pareto efficient if there is no other feasible allocation that makes every user at least as well off and at least one user strictly better off under the capacity constraint.}
\end{subprop}
\begin{proof} Let $\beta\in\mathcal{C}.$ If $$\sum_{j=1}^m \beta^j< C_{\Omega}=\log(1+m\frac{P}{\sigma_0^2})\,,$$ then at least one of the users can improve its rate (hence its payoff) to reach one of the faces of the capacity region. We now check the strategy profile in the face $$ \{ (\alpha^j,\alpha^{-j})\ |\  \alpha^{j}\geq r_m,\sum_{j=1}^m\alpha^{j}=C_{\Omega}\}.$$
If  $$\beta\in \{ (\alpha^j,\alpha^{-j})\ |\  \alpha^{j}\geq r_m,\sum_{j=1}^m\alpha^{j}=C_{\Omega}\},$$ then
from the Lemma\label{lembr}, $\overline{BR}(\beta^{-j})=\{\beta^j\}.$ Hence, $\beta$ is a strict equilibrium. Moreover, this strategy $\beta$ is Pareto optimal because  the rate of each user is maximized under the capacity constraint. These strategies are social welfare if the quantity $$\sum_{j=1}^m u^j(\alpha^j,\alpha^{-j})=\sum_{j=1}^m g(\alpha^j)$$ is maximized.
\end{proof} Note that the set of pure Nash equilibria is a convex subset of the capacity region.

\section{Robust equilibria and efficiency measures} \label{secrobust}

\subsection{Constrained Strong Equilibria and Coalition Proofness}
An action profile  in a local interaction between $m$ senders is a constrained $k-$strong equilibrium if it is feasible and no coalition of size $k$ can
improve the rate transmissions of each of its members with  respect to the capacity constraints. An action is a constrained strong equilibrium \cite{aumann} if it is a constrained $k-$strong equilibrium for any size $k.$  A strong equilibrium is then a policy from which no coalition
(of any size) can deviate and improve the transmission rate of every
member of the coalition (group of the simultaneous moves), while possibly lowering the transmission rate
of users outside the coalition group. This notion of  constrained strong equilibria \footnote{Note that the set of constrained strong equilibria  is a subset of Nash equilibria (by taking coalitions of size one) and any constrained strong equilibrium is Pareto optimal (by taking coalition  of full size).} is very attractive because it is resilient against coalitions of users. Most of the games do not admit any strong
equilibrium but in our case we will show that the multiple access channel game has several strong equilibria.

\begin{subthm} Any rate profile on the maximal  face of the capacity region $\mathcal{C}:$
$$ Face_{\max}(\mathcal{C}):=\{ (\alpha^j,\alpha^{-j})\in\mathbb{R}^m\ |\  \alpha^{j}\geq r_m,\sum_{j=1}^m\alpha^{j}=C_{\Omega}\},$$
is a constrained strong equilibrium.
\end{subthm}
\begin{proof}
 We remark that if the rate profile $\alpha$ is not on the maximal face of the capacity region, then $\alpha$ is not resilient to deviation by a single user. Hence, $\alpha$ cannot be a constrained strong equilibrium. This says that a necessary condition for a rate profile to be a strong equilibrium is to be in the subset  $Face_{\max}(\mathcal{C}).$ We now prove that the condition: $\alpha\in Face_{\max}(\mathcal{C})$ is sufficient. Let $\alpha\in Face_{\max}(\mathcal{C}).$ Suppose that $k$ users deviate simultaneously from the rate profile $\alpha.$ Denote by $Dev$ the set of users which deviate simultaneously (eventually by forming a coalition). The rate constraints of the  deviants are
 \begin{enumerate}
\item ${\alpha'}^j\geq 0, \ \forall j\in Dev,$
\item $\sum_{j\in\bar{J}}{\alpha'}^j\leq C_{\bar{J}},\ \forall \bar{J}\subseteq Dev,$
\item $\sum_{j\in J\cap Dev}{\alpha'}^j\leq C_{J}-\sum_{j\in J, j\notin Dev}\alpha^j$, $\ \forall {J}\subseteq \Omega,\ J\cap Dev\neq \emptyset.$
\end{enumerate}
In particular, for $J=\Omega,$ we have $\sum_{j\in Dev}{\alpha'}^j\leq C_{\Omega}-\sum_{ j\notin Dev}\alpha^j.$ The total rate of the deviants is bounded by
$C_{\Omega}-\sum_{ j\notin Dev}\alpha^j$, which is not controlled by the deviants. The deviants move to $({\alpha'}^j)_{j\in Dev}$ with  $$\sum_{j\in Dev}{\alpha'}^j <C_{\Omega}-\sum_{ j\notin Dev}\alpha^j\,.$$ Then, there  exists $j$ such that $\alpha^j>{\alpha'}^j.$  Since $g$ is non-decreasing, this implies that $g(\alpha^j)>g({\alpha'}^j).$ The user $j$ who is a member of the coalition $Dev$ does not improve its payoff.
If the rates of some of the deviants are increased, then the rates of some other users from coalition must decrease. If $({\alpha'}^j)_{j\in Dev}$ satisfies $$\sum_{j\in Dev}{\alpha'}^j =C_{\Omega}-\sum_{ j\notin Dev}\alpha^j\,,$$ then some users in the coalition $Dev$ have increased their rates compared with $(\alpha^j)_{j\in Dev}$ and some others in $Dev$ have decreased their rates of transmission (because the total rate is the constant $C_{\Omega}-\sum_{ j\notin Dev}\alpha^j).$ The users in $Dev$ with a lower rate ${\alpha'}^j\leq \alpha^j$ do not benefit to be member of the coalition (Shapley criterion of membership of coalition does not hold) . And this holds for any $\emptyset\subsetneqq Dev\subseteqq\Omega.$ This completes the proof.
\end{proof}
\begin{subcoro} In the constrained rate allocation game, Nash equilibria and strong equilibria in pure strategies coincide.
\end{subcoro}
\subsection{Constrained Potential Function for Local Interaction}
Introduce the following function:
$$V(\alpha)=\upharpoonleft_{\mathcal{C}}(\alpha) \sum_{j=1}^m g(\alpha^j)\,,$$ where $\upharpoonleft_{\mathcal{C}}$ is the indicator function of $\mathcal{C}, i.e.,\ $ $\upharpoonleft_{\mathcal{C}}(\alpha)=1$ if $\alpha\in\mathcal{C}$ and $0$ otherwise.
 The function $V$ satisfies $$ V(\alpha)-V(\beta^j,\alpha^{-j})=g(\alpha^j)-g(\beta^j),\ \forall \alpha,(\beta^,\alpha^{-j})\in\mathcal{C} .$$ If $g$ is  differentiable, then one has  $$\frac{\partial}{\partial \alpha^j}V(\alpha)= g'(\alpha^j)=\frac{\partial}{\partial \alpha^j}u^j$$ in the interior of the capacity region $\mathcal{C}$, and $V$ is a constrained potential function \cite{zhu} in pure strategies.

\begin{subcoro}
The local maximizers of $V$ in $\mathcal{C}$ are pure Nash equilibria. Global  maximizers of $V$ in $\mathcal{C}$ are both constrained strong  equilibria and social optima for the local interaction.
\end{subcoro}
\subsection{Strong Price of Anarchy }
Throughout  this subsection, we assume that the function $g$ is the identity function, i.e.,  $g(x)=id(x):=x.$
One of the approaches used to measure how much the performance of decentralized systems is
affected by the selfish behavior of its components is the {\it price of anarchy}. We present a similar price for strong equilibria under the coupled rate constraints. This notion of Price of Anarchy can be seen as an {\it efficiency metric} that measures the {\it price of selfishness} or decentralization and has been extensively used in the context of congestion games or routing games where typically users have to minimize a cost function. In the context of rate allocation in the multiple access channel, we define an equivalent  measure of price of anarchy  for  rate maximization problems.
One of the advantages of a strong equilibrium
is that it has the potential to reduce the distance between the
optimal solution and the solution obtained as an outcome
of selfish behavior, typically  in the case where the capacity constraint is violated at each time.
Since the constrained rate allocation game has  strong equilibria, we can define the  strong price of anarchy, introduced in \cite{stro}, as the ratio between the payoff of the worst constrained
strong equilibrium and  the  social optimum value which $C_{\Omega}$.
\begin{subthm}
The strong price of anarchy  of the  constrained rate allocation game is 1 for $g(x)=x.$
\end{subthm}
%Remark: this result says that the equilibria are efficient for $g=id.$
%\begin{proof} Since the constrained strong equilibria are the rate profiles in the maximal face, they satisfy $ \sum_j \alpha^j=C_{\Omega},$  the constrained strong price of anarchy (CSPoA) is $$\mbox{CSPoA}= \min_{p\in SE}\ \frac{\sum_{j=1}^{N}u^j(\alpha)}{SW}=1,$$ where $SE$ denotes the set of strong equilibria, $SW$ is the total rate obtained at any feasible social welfare value: $SW:=\max_{\alpha}\ \sum_{j}\alpha^j$ such that $\alpha$ satisfies
%$\left\{\alpha\in \R^{m}_{+},\ \sum_{j\in J}\alpha^j\leq  \frac{1}{m} \log\left(1+|J|\frac{P}{\sigma^2_0}\right),\ \forall\ \emptyset \subsetneqq J\subseteq \Omega\right\}.$  This implies that the rate profile at the maximal face of $\mathcal{C}$ gives the social optimum value for $g=id.$ This means that $SW=C_{\Omega}.$
% Hence, ratio between the payoff of the worst
%strong equilibrium and the feasible social welfare value is one.
%\end{proof}
Note that
for $g\neq id,$ the CSPoA can be less than one. However, the  optimistic
price of anarchy of the {\it best constrained  equilibrium}  also called {\it price of stability} \cite{psta} is one for any function $g$ i.e the  efficiency of "best" equilibria is $100\%.$

\section{Selection of Pure Equilibria} \label{secselection}
We have shown in previous sections that our rate allocation game has a continuum of pure Nash equilibria and strong equilibria.  We
address now the problem of selecting one equilibrium which has certain
desirable properties: the normalized pure Nash equilibrium, introduced in \cite{rosen}. See also \cite{corre,cor,ponstein}. We introduce the
Lagrangian that corresponds to the constrained maximization
problem faced by every user  when the other rates are at the maximal face of the polytope $\mathcal{C}$:
\begin{eqnarray}
\max_{\alpha} & u^j(\alpha)\\
 \textrm{s.t. }&\alpha^1+\ldots+\alpha^m=C_{\Omega}
\end{eqnarray}
 and the Lagrangian for user $j$ is given by $$L^j(\alpha,\zeta)=u^j(\alpha)-\zeta^j\left(\sum_{j}\alpha^j-C_{\Omega}\right).$$ From Karush-Kuhn-Tucker optimality conditions, it follows that there exists $\zeta\in\mathbb{R}^m$ such that $$ g'(\alpha^j)=\zeta^j,\ \sum_{j=1}^m\alpha^j=C_{\Omega}.$$ For a fixed vector $\zeta$ with identical entries, define the normal form game $\Gamma({\zeta})$ with $m$ users, where actions are taken as rates and the payoffs  given by $L(\alpha,\zeta).$ A normalized equilibrium is an equilibrium of the game $\Gamma(\zeta^*)$ where $\zeta^*$ is normalized into the form ${\zeta^*}^j=\frac{c}{\tau^j},\ c>0,\tau^j>0.$
We now have the following result due to Goodman \cite{cor} which implies Rosen's condition on uniqueness for strict concave games.
\begin{subthm} \label{ret1}
Let $u^j$ be a smooth and strictly concave function in $\alpha^j,$ each $u^j$ be convex in $\alpha^{-j}$, and there exist some $\zeta$ such that the weighted non-negative sum of the payoffs $
\sum_{j=1}^m\zeta^j u^j(\alpha)
$ is concave in $\alpha.$ Then the matrix $$G(\alpha,\zeta)+G^{T}(\alpha,\zeta)$$ is  negative definite (which implies uniqueness) where $G(\alpha,\zeta)$ is the Jacobian with respect to $\alpha$ of
 $$h(\alpha,\zeta):=\left[\zeta^1 \nabla_1 u^1(\alpha), \zeta^2 \nabla_2 u^2(\alpha),\ldots,
 \zeta^{m} \nabla_{m} u^{m}(\alpha)
 \right]^T$$
and $G^{T}$ is the transpose of the matrix $G.$
\end{subthm}
This now leads to the following corollary for our problem.
\begin{subcoro}
If $g$ is a non-decreasing strictly concave function, then the rate allocation game has a unique normalized equilibrium which corresponds to an equilibrium of the normal form game with payoff $L(\alpha,\zeta^*)$ for some $\zeta^*.$
\end{subcoro}

\section{Stability and Dynamics} \label{secdynamics}
In this section, we study the stability of equilibria and several classes of evolutionary game dynamics. We show that the evolutionary game has a unique pure constrained evolutionary stable strategy.
\begin{prop}
The collection of rates $$\alpha=\left(\frac{C_{\Omega}}{m},\ldots,\frac{C_{\Omega}}{m}\right)\,,$$ i.e the distribution of Dirac concentrated on the rate $\frac{C_{\Omega}}{m},$ is the unique  symmetric pure Nash equilibrium.
\end{prop}
\begin{proof}
Since the constrained rate allocation game is symmetric, there exists a symmetric  (pure or mixed) Nash equilibrium. If such an equilibrium exists in pure strategies, each user transmits with the same rate $r^*.$ It follows from Proposition \ref{ne}, and the bound $r_m\leq \frac{C_{\Omega}}{m}$ that $ r^*$ satisfies $m r^*=C_{\Omega}$  and $r^*$ is feasible.
\end{proof}

Since the set of feasible actions is convex, we can define convex combination of rates in the set of the feasible rates. For example, $\epsilon \alpha'+(1-\epsilon)\alpha$ is a feasible rate if $\alpha'$ and $\alpha$ are feasible. The symmetric rate profile  $(r, r,\ldots,r)$ is feasible if and only if $0\leq r\leq r^*=\frac{C_{\Omega}}{m}.$  We say that the rate $r$ is a constrained evolutionary stable strategy (ESS) if it is feasible and for every {\it mutant strategy} $mut\neq \alpha$ there exists $\epsilon_{mut}>0$ such that
$$\left\{\begin{array}{cc}
r_{\epsilon}:=\epsilon\ mut+(1-\epsilon)r\in \mathcal{C} & \forall \epsilon\in(0,\epsilon_{mut})\\
u(r,r_{\epsilon},\ldots,r_{\epsilon})>u(mut, r_{\epsilon},\ldots,r_{\epsilon}) & \forall \epsilon\in(0,\epsilon_{mut})
\end{array}\right.$$

\begin{subthm} The pure strategy $r^*=\frac{C_{\Omega}}{m}$ is a constrained evolutionary stable strategy.
\end{subthm}
\begin{proof} Let $mut\leq r^*$
The rate $\epsilon\ mut+(1-\epsilon)r^*$ is feasible implies that $mut\leq r^*$ (because $r^*$ is the maximum symmetric rate achievable). Since $mut\neq r^*,$ $mut$ is strictly lower than $r^*.$ By monotonicity of the function $g,$ one has $$ u(r^*,\epsilon\ mut+(1-\epsilon)r^*)>u(mut,\epsilon\ mut+(1-\epsilon)r^*),\ \forall \epsilon.$$ This completes the proof.
\end{proof}

\subsection{Symmetric Mixed Strategies}
Define the mixed capacity region $\mathcal{M}(\mathcal{C})$ as the set of measures profile $(\mu^1,\mu^2,\ldots,\mu^m)$ such that $$\int_{\mathbb{R}_{+}^{|J|}}\left(\sum_{j\in J}\alpha^j\right)\prod_{j\in J}\mu^j(d\alpha^j) \leq C_{J},\ \forall J\subseteq 2^{\Omega}.$$
Then the payoff of the action $a\in\mathbb{R}_{+}$ satisfying $(a,\lambda,\ldots,\lambda)\in \mathcal{M}(\mathcal{C})$  can be defined as $$ F(a,\mu)=\int_{[0,\infty[^{m-1}} u(a,b_2,\ldots, b_m)\ \nu_{m-1}(db)\,,$$ where $\nu_k=\bigotimes_{1}^{k} \mu$ is the product measure on $[0,\infty[^{k}.$
The constraint set becomes the set of probability measures on $\mathcal{R}_+$ such that  $$ 0\leq \mathbb{E}(\mu):=\int_{\mathbb{R}_+}\ \alpha^j\ \mu(d\alpha^j)\leq \frac{C_{\Omega}}{m}<C_{\{1\}}\,.$$

\begin{lem}
 $F(a,\mu)=\upharpoonleft_{[0,C_{\Omega}-(m-1)\mathbb{E}(\mu)]} \times g(a)\times$ $$ \int_{b\in \mathcal{D}_a} \ \nu_{m-1}(db)=\upharpoonleft_{[0,C_{\Omega}-(m-1)\mathbb{E}(\mu)]}\times g(a)\nu_{m-1}(\mathcal{D}_a)$$
 \newline where $$ \mathcal{D}_a=\{(b_2,\ldots, b_m)\ |\  (a,b_2,\ldots, b_m)\in \mathcal{C} \}\,.$$
\end{lem}
\begin{proof} If the rate does not satisfy the capacity constraints, then the payoff is $0.$ Hence the {\it rational} rate for user $j$ is lower than $C_{\{j\}}.$ Fix a rate $a\in[0,C_{\{j\}}].$
Let $D^a_J :=C_J-a\delta_{\{1\in J\}}.$ Then, a necessary condition to have a non-zero payoff is  $$(b_2,\ldots, b_m) \in \mathcal{D}_a\,,$$ where $$\mathcal{D}_a=\{(b_2,\ldots, b_m)\in \mathbb{R}_{+}^{m-1},\ \sum_{j\in J,j\neq 1}b_j\leq D^a_J,\ J\subseteq 2^{\Omega} \}.$$
Thus,
\begin{eqnarray} \nonumber F(a,\mu)&=&\int_{\mathbb{R}_{+}^{m-1}} u(a,b_2,\ldots, b_m)\ \nu_{m-1}(db)\\ \nonumber & =&\int_{b\in \mathbb{R}_{+}^{m-1},\ (a,b)\in\mathcal{C}} g(a)\ \nu_{m-1}(db)\\
\nonumber &=& \upharpoonleft_{[0,C_{\Omega}-(m-1)\mathbb{E}(\mu)]} g(a)\\
\nonumber&&\times \int_{b\in \mathcal{D}_a} \ \nu_{m-1}(db)
\end{eqnarray}
\end{proof}

   \subsection{Constrained Evolutionary Game Dynamics}
   The class of evolutionary games in large population provides a simple
framework for describing strategic interactions among large numbers of users.  In this subsection we  turn to modeling the behavior of the users who play them.
Traditionally, predictions of behavior in game theory are based on some notion of equilibrium,
typically Cournot equilibrium, Bertrand equilibrium, Nash equilibrium, Stackelberg solution, Wardrop equilibrium or some refinement thereof. These notions require the assumption of equilibrium knowledge, which posits that each user correctly anticipates
how his opponents will act. The equilibrium knowledge assumption is too strong and is difficult
to justify in particular in contexts with large numbers of users.
As an alternative to the equilibrium approach, we propose an explicitly dynamic
 updating choice, a procedure in which users myopically update their behavior in response to
their current strategic environment. This dynamic procedure does not assume the automatic
coordination of users' actions and beliefs, and it can derive many specifications of users'
choice procedures.
These procedures are specified formally by defining a revision of rates called {\it revision
protocol} \cite{sandholmbook}.  A revision
protocol takes current payoffs  and current mean rate and maps to conditional switch rates which describe how frequently users in some class playing rate $\alpha$  who
are considering switching rates switch to strategy $\alpha'.$ Revision protocols are flexible enough to
incorporate a wide variety of paradigms, including ones based on imitation, adaptation, learning, optimization, etc.

 We  use a class of continuous evolutionary dynamics.  We refer to \cite{wiopt,gamecomm,mass} for evolutionary game dynamics with or  without time delays. The continuous-time evolutionary game
dynamics on the measure space $(\mathcal{A}, \mathcal{B}(\mathcal{A}),\mu)$ is  given by
 \begin{equation} \dot{\lambda}_t(E)= \int_{a\in E}V(a,\lambda_t) \mu(da) \end{equation}
where $$V(a,\lambda_t)=K\left[\int_{x\in \mathcal{A}}\beta^x_a(\lambda_t)\lambda_t(dx)-\int_{x\in \mathcal{A}}\beta^a_x(\lambda_t) \lambda_t(dx)\right],$$ and $\beta^x_a$ represents the rate of mutation from $x$ to $a,$ and $K$ is a growth parameter.  $\beta^x_a(\lambda_t)=0$ if $(x,\lambda_t)$ or $(a,\lambda_t)$ is not in the (mixed) capacity region, $E$ is a $\mu-$measurable subset of $ \mathcal{A}.$
At each time $t,$ probability measure $\lambda_t$ satisfies
$\frac{d}{dt}\lambda_t(\mathcal{A})=0$.
\medskip

\noindent{\bf Constrained Brown-von Neumann-Nash dynamics.}
\newline The constrained revision protocol is $$\beta^x_a(\lambda_t)=\left\{\begin{array}{c} \max(F(a,\lambda_t)-\int_x F(x,\lambda_t)\ dx, 0) \\ \mbox{if}\ (a,\lambda_t),\ (x,\lambda_t)\in \mathcal{M}(\mathcal{C}) \\ 0 \ \mbox{otherwise}
\end{array}\right.$$
%The constrained BNN dynamics becomes
%$$\dot{\lambda}_t(E)=\int_{a | (a,\lambda_t)\in\mathcal{M}(\mathcal{C})}\ \max[0,F(a,\lambda_t)-\int_x F(x,\lambda_t)] \mu(da)-\mu(E)\int_{y\in\mathcal{A}, (y,\lambda_t)\in \mathcal{M}(\mathcal{C})} \ \max[F(y,\lambda_t)-\int_x F(x,\lambda_t),0]  $$
%
%We can easily verify this for finite discrete action space $\mathcal{A}.$ This dynamics can be rewritten as
%$$\dot{\lambda}_t(e)=\sum_{a\in\mathcal{A},\ (a,\lambda_t)\in\mathcal{M}(\mathcal{C})}\beta^a_e(\lambda_t)\lambda_t(a) -\lambda_t(e)\sum_{a\in\mathcal{A}}\beta^e_a(\lambda_t)$$
%
\noindent{\bf Constrained Replicator Dynamics.}
$$\beta^x_a(\lambda_t)=\left\{\begin{array}{c} \max(F(a,\lambda_t)-F(x,\lambda_t), 0) \\ \mbox{if}\ (a,\lambda_t),\ (x,\lambda_t)\in \mathcal{M}(\mathcal{C})\\ 0 \  \mbox{otherwise}\end{array}\right.$$
\noindent{\bf Constrained $\theta-$Smith Dynamics.}
$$ \beta^x_a(\lambda_t)=\left\{\begin{array}{cc} \max(F(a,\lambda_t)-F(x,\lambda_t), 0)^{\theta} \\ \mbox{if}\ (a,\lambda_t),\ (x,\lambda_t)\in \mathcal{M}(\mathcal{C})\\ 0 \ \mbox{otherwise}\end{array}\right.,\ \theta\geq 1$$
\medskip

We now provide a common property that applies to all these dynamics:  the set of Nash equilibria is a subset of rest points (stationary points) of the evolutionary game dynamics. Here we extend to evolutionary game with a continuous action space and coupled constraints, and more than two-users interactions. The counterparts of these results in discrete action space can be found in \cite{hofbauer,sandholmbook}.

\begin{subthm} Any  Nash equilibrium of the game is a rest point of the following evolutionary game dynamics: constrained Brown-von Neumann-Nash, generalized Smith dynamics, and replicator dynamics. In particular, the evolutionary stable strategies set  is a subset of the rest points of these constrained evolutionary game dynamics.
\end{subthm}

\begin{proof}
It is clear for pure equilibria by using the revision protocols $\beta$ of these dynamics. Let $\lambda$ be an equilibrium. For any rate $a$ in the support of $\lambda,$ $\beta^a_x=0$ if $F(x,\lambda)\leq F(a,\lambda).$ Thus, if $\lambda$ is an equilibrium the difference between the microscopic inflow and outflow is $V(a,\lambda)=0$, given that $a$ is the support of the measure $\lambda.$
\end{proof}

Let $\lambda$ be a  finite Borel measure on $[0,C_{\{j\}}]$ with full support. Suppose  $g$ is continuous on $[0,C_{\{j\}}].$ Then,
$\lambda$ is a rest point of the BNN dynamics if and only if $\lambda$ is a symmetric
Nash equilibrium.
Note that the choice of topology is an important issue when defining dynamics
convergence and stability. The most used in this area is the topology of the weak
convergence to measure closeness of two states of the system. Different
distances (Prohorov metric, metric on bounded and Lipschitz continuous
functions on $\mathcal{A}$) have been proposed. We refer the reader
to \cite{pierre}, and the references therein for more
details on {\it evolutionary robust strategy} and stability notions.

\section{Generalization} In this section, we consider the asymmetric case. Each user has its maximum power $P_i$ and a channel gain $h_i.$ In addition, the rate of transmission is subject to a coupled capacity constraint. The capacity region $\mathcal{C}$  is described by the set
\begin{equation}\label{capacityregion}
\left\{\alpha\in \mathbb{R}^{m}_{+}, \sum_{i\in \Omega}\alpha^i\leq  C_\Omega,\ \forall\ \emptyset \subset \Omega \subseteq \mathcal{N}\right\},
\end{equation}
where $\Omega$ is any subset of $\mathcal{N}$ and  \begin{equation}\label{OmegaCapacity} C_\Omega =\log\left(1+\sum_{i\in \Omega}\frac{P_i h_i}{\sigma^2_0}\right),\end{equation}
is the capacity for users in $\Omega$. The capacity region reveals a competitive nature of the interactions among senders: if a user $i$ wants to communicate at a higher rate, one of the other users  has to lower his rate;  otherwise, the capacity constraint is violated.  We let $$r_{i,\Omega} :=\log\left(1+\frac{P_ih_i}{\sigma_0^2+\sum_{i'\in \Omega,i'\neq i} P_{i'}h_{i'}}\right)$$  denote the  bound rate  of a user when the signals of the $|\Omega|-1$ other users are treated as  noise.

Due to the noncooperative nature of the rate allocation, we can formulate the one-shot game $$\Xi =\langle \mathcal{N}, (\mathcal{A}^i)_{i\in\mathcal{N}}, (u^i)_{i\in\mathcal{N}}\rangle\,,$$ where the set of users $\mathcal{N}$ is the set of players, $\mathcal{A}^i$, $i\in  \mathcal{N}$, is  the set of actions, and $u^i$, $i\in  \mathcal{N}$, are the payoff functions.
We define $u^i:\prod_{i=1}^m\mathcal{A}^i\rightarrow \mathbb{R}_+$ as follows.
\begin{eqnarray}
u^i(\alpha^{i},\alpha^{-i})&=&\upharpoonleft_{\mathcal{C}}(\alpha)g^i(\alpha^i,\alpha^{-i})\\ &=&\left\{ \begin{array}{ll} g^i(\alpha^i) & {\textrm{~if~}}\   (\alpha^{i},\alpha^{-i})\in \mathcal{C}\\
0 & \mbox{otherwise}
\end{array}
\right.,
\end{eqnarray}
 where  $\upharpoonleft_{\mathcal{C}}$ is the indicator function; $\alpha^{-i}$ is a vector consisting of other players' rates, i.e., $\alpha^{-i}=[\alpha^1,\ldots,\alpha^{i-1},\alpha^{i+1},\ldots,\alpha^N]$  and $u^i$
 %:\ \R_+ \rightarrow \R_+$
 is a positive and strictly increasing  function for each fixed $\alpha^{-i}$. Since the game is subject to coupled constraints, the action set $\mathcal{A}^i$ is coupled and dependent on other players' actions. Given the strategy profile $\alpha^{-i}$ of other players, the constrained action set $\mathcal{A}^i$ is given by
 \begin{equation}
 \mathcal{A}^i(\alpha^{-i}):=\{\alpha^i\in [0,C_{\{i\}}],\ (\alpha^i,\alpha^{-i})\in \mathcal{C} \}
 \end{equation}  We then have an asymmetric game. The minimum rate that the user $i$ can guarantee in the feasible regions is $r_{i,\mathcal{N}}$ which is different than $r_{j,\mathcal{N}}.$

 Each user $i$ maximizes $u^i(\alpha^{i},\alpha^{-i})$ over the coupled constraint set.
Owing to the monotonicity of the function $g^i$ and the inequalities that define the capacity region, we obtain the following lemma.
\begin{lem}
 Let $\overline{BR}^i(\alpha^{-i})$ be the best reply to the strategy $\alpha^{-i}$, defined by  $$ \overline{BR}^i(\alpha^{-i})=\arg\max_{y\in \mathcal{A}^i(\alpha^{-i})}u^i(y,\alpha^{-i}).$$ $\overline{BR}^i$ is a non-empty single-valued correspondence (i.e a  standard function), and is given by
\begin{equation}
\max \left(r_{i,\mathcal{N}},\min_{\Omega\in\Gamma_i}\left\{C_{\Omega}-\sum_{k\in\Omega\backslash\{i\}}\alpha^k\right\} \right),
\end{equation} \label{lembr1}
where $\Gamma_i=\{\Omega\in 2^{\mathcal{N}}, i\in\Omega\}$.
\end{lem}
\begin{prop}
 The set of Nash equilibria is $$\{ (\alpha^i,\alpha^{-i})\ |\  \alpha^{i}\geq r_{i,\mathcal{N}},\sum_{i\in\mathcal{N}}\alpha^{i}=C_{\mathcal{N}}\}.$$
All these equilibria are optimal in Pareto sense.
\end{prop}
\begin{proof}
Let $\beta$ be a feasible solution, i.e., $\beta\in\mathcal{C}.$ If $$\sum_{i=1}^N \beta^i< C_{\mathcal{N}}=\log\left(1+\sum_{i\in \mathcal{N}} \frac{P_ih_i}{\sigma_0^2}\right),$$ then at least one of the users can improve its rate (hence its payoff) to reach one of the faces of the capacity region. We now check the strategy profile on the face
$$ \left\{ (\alpha^i,\alpha^{-i})\ \bigg|\  \alpha^{i}\geq r_{i,\mathcal{N}},\sum_{i=1}^N\alpha^{i}=C_{\mathcal{N}}\right\}.$$
If  $$\beta\in \left\{ (\alpha^i,\alpha^{-i})\ \bigg|\  \alpha_{i}\geq r_{i,\mathcal{N}},\sum_{i=1}^N\alpha^{i}=C_{\Omega}\right\},$$ then
from the Lemma \ref{lembr1}, $\overline{BR}^i(\beta^{-i})=\{\beta^i\}.$ Hence, $\beta$ is a strict equilibrium. Moreover, this strategy $\beta$ is Pareto optimal because  the rate of each user is maximized under the capacity constraint. These strategies are social welfare optimal if the total utility $$\sum_{i=1}^N u^i(\alpha^i,\alpha^{-i})=\sum_{i=1}^N g^i(\alpha^i)$$ is maximized subject to constraints.
\end{proof}

Note that the set of pure Nash equilibria is a convex subset of the capacity region.
The pure equilibria are global optima\footnote{This implies that the price of anarchy is one. } if the function $g$ is the identity function.

\section{Concluding remarks} \label{secconclud}
 In this paper, we have studied an evolutionary  Multiple Access Channel game with a continuum action space and coupled rate constraints. We showed that the  game has a continuum of strong equilibria which  are 100\% efficient in the rate optimization problem.  We proposed the constrained Brown-von Neumann-Nash dynamics, Smith dynamics, and the replicator dynamics  to study the stability of equilibria in the  long run. An interesting question which we leave for future work is
whether similar equilibria structure exist in the case of multiple access games with non-convex capacity regions. Another extension would be to the hybrid model in which users can select among several receivers and control the total rate, which is currently under study.

\end{document}